\documentclass[oneside,british]{amsart}
\usepackage[T1]{fontenc}
\usepackage[latin9]{inputenc}
\usepackage{textcomp}
\usepackage{amsthm}
\usepackage{amstext}
\usepackage{amssymb}
\usepackage{graphicx}
\usepackage{esint}

\makeatletter
\numberwithin{equation}{section}
\numberwithin{figure}{section}
\theoremstyle{plain}
\newtheorem{thm}{\protect\theoremname}
  \theoremstyle{plain}
  \newtheorem{prop}[thm]{\protect\propositionname}
  \theoremstyle{definition}
  \newtheorem{example}[thm]{\protect\examplename}
  \theoremstyle{definition}
  \newtheorem{defn}[thm]{\protect\definitionname}
  \theoremstyle{plain}
  \newtheorem{lem}[thm]{\protect\lemmaname}
  \theoremstyle{plain}
  \newtheorem{cor}[thm]{\protect\corollaryname}

\makeatother

\usepackage{babel}
  \providecommand{\corollaryname}{Corollary}
  \providecommand{\definitionname}{Definition}
  \providecommand{\examplename}{Example}
  \providecommand{\lemmaname}{Lemma}
  \providecommand{\propositionname}{Proposition}
\providecommand{\theoremname}{Theorem}

\begin{document}

\title{Maxwell\textasciiacute{}s Equations, The Euler index and Morse theory}

\author{Carlos Valero \\
Universidad Autonoma Metropolitana (UAM)\\
Unidad Cuajimalpa\\
Mexico City, Mexico\\
\\
}

\maketitle
\global\long\def\CC{\mathbb{C}}

\global\long\def\RR{\mathbb{R}}

\global\long\def\map{\rightarrow}

\global\long\def\mult{\mathcal{M}}

\global\long\def\EE{\mathcal{E}}

\global\long\def\OO{\mathcal{O}}

\global\long\def\FH{\mathcal{F}}

\global\long\def\CH{\mathcal{H}}

\global\long\def\VV{\mathcal{V}}

\global\long\def\PP{\mathcal{P}}

\global\long\def\CS{\mathcal{C}}

\global\long\def\tangent{T}

\global\long\def\cotangent{\tangent^{*}}

\global\long\def\conormal{\mathcal{C}}

\global\long\def\SS{\mathcal{S}}

\global\long\def\KK{\mathcal{K}}

\global\long\def\NN{\mathcal{N}}

\global\long\def\sym#1{\hbox{S}^{2}#1}

\global\long\def\symz#1{\hbox{S}_{0}^{2}#1}

\global\long\def\proj#1{\hbox{P}#1}

\global\long\def\SO{\hbox{SO}}

\global\long\def\GL{\hbox{GL}}

\global\long\def\U{\hbox{U}}

\global\long\def\tr{\hbox{tr}}

\global\long\def\ZZ{\mathbb{Z}}

\global\long\def\der#1#2{\frac{\partial#1}{\partial#2}}

\global\long\def\covder{\hbox{D}}

\global\long\def\diff{d}

\global\long\def\dero#1{\frac{\partial}{\partial#1}}

\global\long\def\ind{\hbox{\, ind}}

\global\long\def\deg{\hbox{deg}}

\global\long\def\smb{S}

\begin{abstract}
We show show that the singularities of the Fresnel surface for Maxwell\textasciiacute{}s
equation on an anisotrpic material can be accounted from purely topological
considerations. The importance of these singularities is that they
explain the phenomenon of conical refraction predicted by Hamilton.
We show how to desingularise the Fresnel surface, which will allow
us to use Morse theory to find lower bounds for the number of critical
wave velocities inside the material under consideration. Finally,
we propose a program to generalise the results obtained to the general
case of hyperbolic differential operators on differentiable bundles.
\end{abstract}

\section{Introduction}

One of the most interesting problems in geometrical optics in the
nineteenth century was the phenomenon of double refraction, in which
a ray of light entering certain crystals is refracted into two rays.
Later on, Hamilton discovered that if the ray of light enters a biaxial
crystal in certain directions (known as the optical axis) then the
ray of light must be refracted into a cone of rays. This phenomenon,
now known as conical refraction, was confirmed a year later by Lloyd,
who investigated aragonite at Hamilton's suggestion.

In mathematical terms, conical refraction is explained by the existence
of conical singularities in the Fresnel surface associated to the
crystal. The Fresnel surface is constructed as the set of allowed
wave speeds imposed by Maxwell's equations. In this paper we will
show that the singularities of the Fresnel surface can be accounted
for topological considerations only. Furthermore, we apply Morse theory
to establish a connection between the number of these singularities
and the critical speeds of wave propagation within the crystal. Inspired
by this result, we propose a program to study this problem in the
general context of hyperbolic differential equations on manifolds.

Recall that Maxwell's equations in a medium with dielectric tensor
$\epsilon\in\sym{\RR^{3}}$ are given by (see \cite[pg. 678]{kn:born})
\begin{eqnarray*}
\nabla\times H-\dero t(\epsilon E) & = & 0,\\
\nabla\times E+\dero tH & = & 0,\\
\nabla\cdot(\epsilon E) & = & 0,\\
\nabla\cdot H & = & 0,
\end{eqnarray*}
where we have assumed the speed of light to be equal to one. These
equations describe the behaviour of the electric and magnetic fields
$E$ and $H$, in the medium under consideration. From now on we will
always assume that the eigenvalues of $\epsilon$ are all strictly
positive real numbers. If for constant vectors $E_{0},H_{0},\xi\in\RR^{3}$
and $\tau\in\RR$ we try to find planar wave solutions of the form
\begin{eqnarray*}
E(x,t) & = & E_{0}\exp(i(<\text{\ensuremath{\xi}},x>-\tau t)),\\
H(x,t) & = & H_{0}\exp(i(<\xi,x>-\tau t)),
\end{eqnarray*}
we are lead to the equations 
\begin{equation}
Q_{\epsilon}(\xi)\left(\begin{array}{c}
E_{0}\\
H_{0}
\end{array}\right)=0,<\xi,\epsilon E_{0}>=0,<\xi,H_{0}>=0,\label{eq:PolarisationMaxwell}
\end{equation}
where
\begin{equation}
Q_{\epsilon}(\xi,\tau)=\left(\left(\begin{array}{cc}
\epsilon & 0\\
0 & I
\end{array}\right)\tau+\left(\begin{array}{cc}
0 & P(\xi)\\
-P(\xi) & 0
\end{array}\right)\right)\label{eq:SymbolEpsilon}
\end{equation}
and $P(\xi)$ is the anti-symmetric matrix such that $P(\xi)v=\xi\times v$
for any $v\in\RR^{3}$. For the first of the equations \ref{eq:PolarisationMaxwell}
to hold we need that
\begin{equation}
\det(Q_{\epsilon}(\xi,\tau))=0.\label{eq:EikonalEquation}
\end{equation}
We can assume without loss of generality that the principal axes of
$\epsilon$ are aligned with the $x$,$y$ and $z$ axes in $\RR^{3}$,
so that we can write 
\[
\epsilon=\left(\begin{array}{ccc}
\epsilon_{1} & 0 & 0\\
0 & \epsilon_{2} & 0\\
0 & 0 & \epsilon_{3}
\end{array}\right).
\]
A simple calculation shows that if we let $q_{\epsilon}(\xi,\tau)=\det(Q_{\epsilon}(\xi,\tau))/\tau^{2}$
(assuming $\tau\not=0)$ then we have that
\begin{eqnarray*}
q_{\epsilon}(\xi,\tau) & = & (\xi_{1}^{2}+\xi_{2}^{2}+\xi_{3}^{3})(\epsilon_{1}\xi_{1}^{2}+\epsilon_{2}\xi_{2}^{2}+\epsilon_{3}\xi_{3}^{3})+\\
 & + & (\epsilon_{1}\epsilon_{2}(\xi_{1}^{2}+\xi_{2}^{2})+\epsilon_{2}\epsilon_{3}(\xi_{2}^{2}+\xi_{3}^{2})+\epsilon_{1}\epsilon_{3}(\xi_{1}^{2}+\xi_{3}^{2}))\tau^{2}\\
 & + & \epsilon_{1}\epsilon_{2}\epsilon_{3}\tau^{4}
\end{eqnarray*}
For $\xi\not=0$ and $\tau\not=0$ we can visualise condition \ref{eq:EikonalEquation}
if we consider the Fresnel surface 
\begin{equation}
\FH_{\epsilon}=\{(\tau\xi/||\xi||^{2})\in\RR^{3}|q_{\epsilon}(\text{\ensuremath{\xi},\ensuremath{\tau})}=0\}.\label{eq:FresnelSurface}
\end{equation}
Physically, the Fresnel surface is the space of the allowable \emph{phase
velocities}. We can visualise $\FH_{\epsilon}$ by considering the
following cases:
\begin{enumerate}
\item If $\epsilon_{1}=\epsilon_{2}=\epsilon_{3}\not=0$ then
\[
q_{\epsilon}(\xi,\tau)=\epsilon_{1}(\xi_{1}^{2}+\xi_{2}^{2}+\xi_{3}^{2}-\epsilon_{1}\tau^{2}),
\]
and hence $\FH_{\epsilon}$ is a sphere of radius $\epsilon_{1}^{-1/2}$.
\item If exactly two if the $\epsilon_{i}$\textasciiacute{}s are equal
then $\FH_{\epsilon}$ consists of two smooth surfaces that intersect
tangentially (see Figure \ref{fig:Uniaxial}).
\item If all of $\epsilon_{i}$\textasciiacute{}s are different from each
other, then $\FH_{\epsilon}$ consists two singular surfaces (each
having 4 singularities) intersecting at these singular points (see
Figure \ref{fig:Biaxial}).
\end{enumerate}
The existence of singularities in case $3$ accounts for the phenomena
of conical refraction described by Hamilton (see \cite[pg. 681]{kn:born})
in which ray of light splits into a cone upon entering at certain
directions of a biaxial crystal (e.g aragonite). 

\begin{figure}

\includegraphics[scale=0.1]{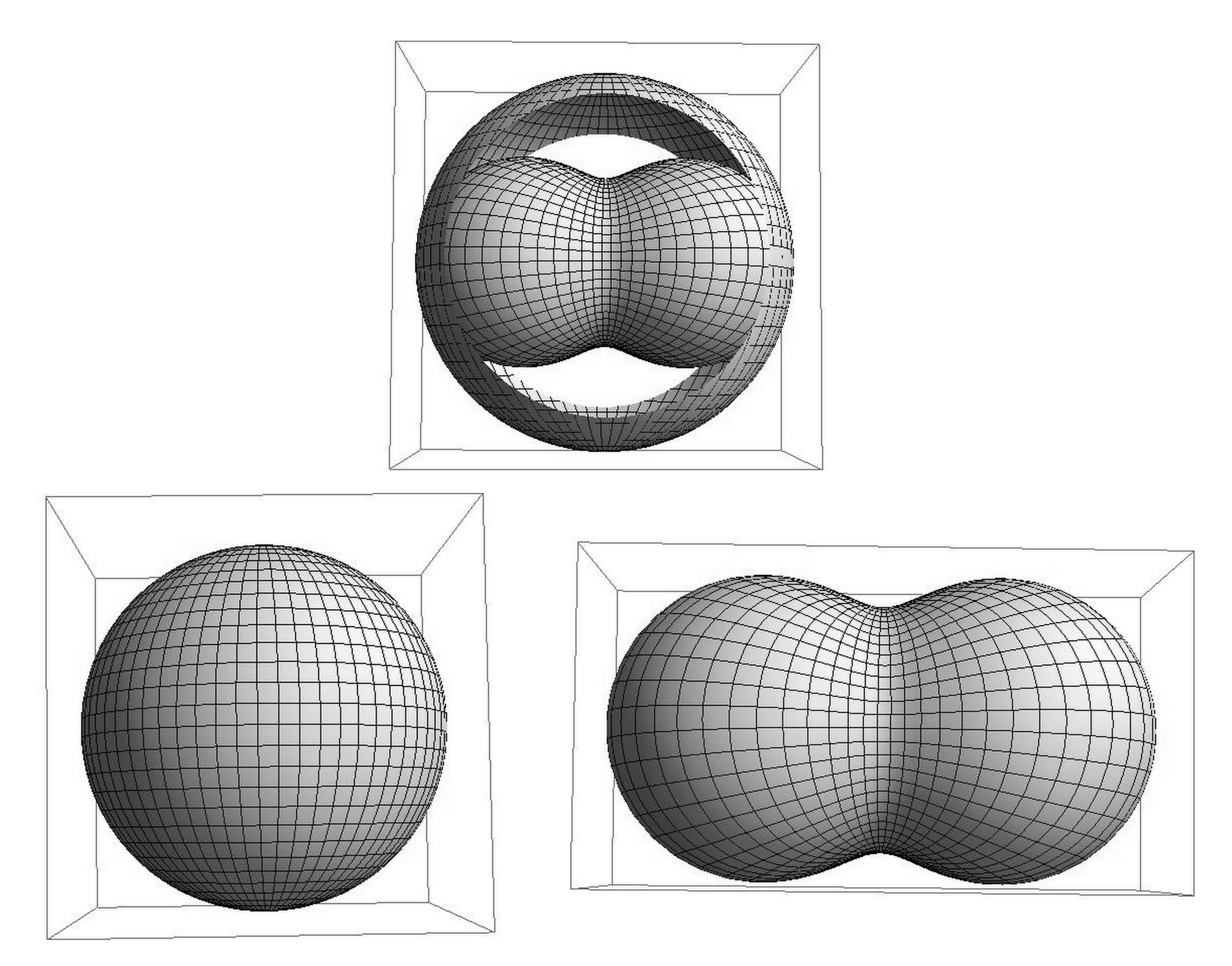}\caption{\label{fig:Uniaxial}Fresnel surface for $\epsilon_{1}=\epsilon_{3}=3$
and $\epsilon_{2}=15$}

\end{figure}

\begin{figure}

\includegraphics[scale=0.1]{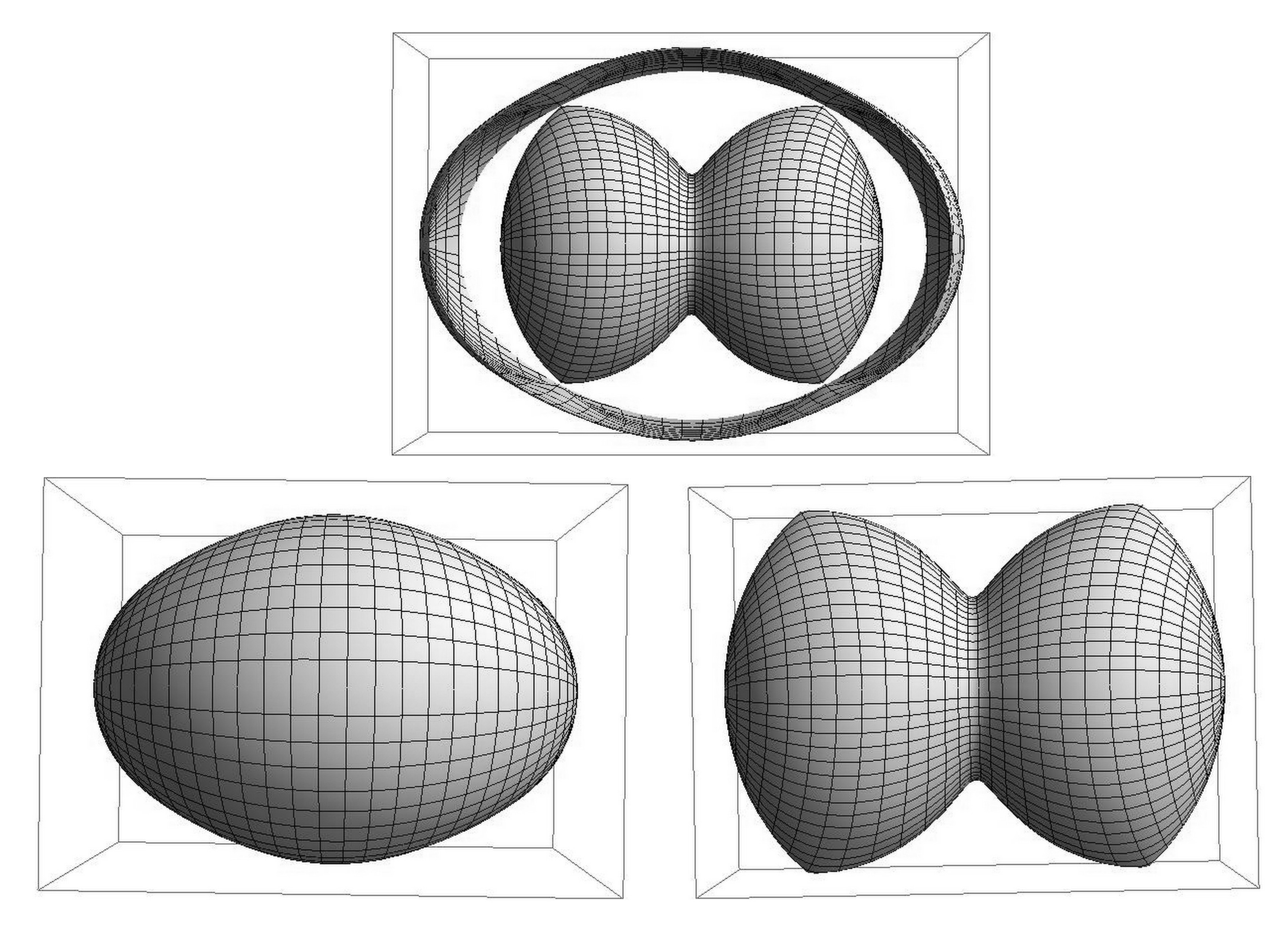}\caption{\label{fig:Biaxial}Fresnel Surface for $\epsilon_{1}=1,\epsilon_{2}=2$
and $\epsilon_{3}=15$}

\end{figure}

\section{Maxwell Equations and Symmetric Tensors on the Sphere}

In this section we will show that, for the case when the eigenvalues
of the dielectric tensor are all different, the singularities of the
associated Fresnel surface can be accounted purely from topological
considerations.

Throughout this section we will assume $F$ to be an orientable Riemannian
real vector bundle of rank $2$ over an orientable closed surface
$X$, and we will denote the metric in $F$ by $<,>_{F}$. We will
let $\sym F$ stand for the bundle of symmetric endomorphism of $F$,
i.e for $p\in X$ we have that $A\in\sym{F_{p}}$ iff $A$ is a linear
map from $F_{p}$ to $F_{p}$ such that
\[
<Av,w>_{F}=<v,Aw>_{F}\hbox{\, for\, all\,\,}v,w\in F_{p},
\]
For a section of $s$ of $\sym F$ we can consider its eigenvalue
functions $\lambda_{s,1},\lambda_{s,2}:X\rightarrow\RR$, i.e for
$p\in X$ we will let $\lambda_{s,1}(p)$ and $\lambda_{s,1}(p)$
be the two eigenvalues of $s(p)$ for all $p\in X$. From now on we
will assume that we have chosen $\lambda_{s,1}$ and $\lambda_{s,2}$
so that $\lambda_{s,1}\leq\lambda_{s,2}$.

Given a symmetric matrix $\epsilon\in\sym{\RR^{3}}$ we define $s_{\epsilon}:S^{2}\rightarrow\sym{(\tangent S^{2})}$
by letting 
\[
<s_{\epsilon}(p)v,w>=<\epsilon v,w>\hbox{\,\, for\, all\,\,\,}p\in S^{k-1}\hbox{\,\, and\,\,}v,w\in\tangent_{p}S^{2}.
\]
Observe that we are identifying $\tangent_{p}S^{2}$ with the subspace
of $\RR^{3}$ consisting of all the vectors orthogonal to $p$, and
that the metric we are using on $\tangent_{p}S^{2}$ is that induced
from the ambient space $\RR^{3}$. To simplify notation we will let
\[
\lambda_{\epsilon,i}=\lambda_{s_{\epsilon},i}\hbox{\,\, for\,\,}i=1,2.
\]
The following result expresses the Fresnel surface $\FH_{\epsilon}$
in terms of $\lambda_{\epsilon,1}$ and $\lambda_{\epsilon,2}$.
\begin{prop}
If $\FH_{\epsilon}$ is the Fresnel surface associated with the dielectric
tensor $\epsilon\in\sym{\RR^{3}}$, then we have that
\[
\FH_{\epsilon}=\bigcup_{i=1}^{2}\{\lambda_{\epsilon,i}^{-1/2}(\xi)\xi|\xi\in S^{2}\}.
\]
\end{prop}
\begin{proof}
As we have seen, we must have that
\[
Q_{\epsilon}(\xi,\tau)\left(\begin{array}{c}
E_{0}\\
H_{0}
\end{array}\right)=0,
\]
where $Q_{\epsilon}$ is given by \ref{eq:SymbolEpsilon}. The above
formula is the same as the formulas
\begin{eqnarray*}
\tau\epsilon E_{0}+\xi\times H_{0} & = & 0,\\
\tau H_{0}-\xi\times E_{0} & = & 0.
\end{eqnarray*}
From these last two equations we obtain that
\[
(\tau/||\xi||)^{2}\epsilon E_{0}-(\xi/||\xi||)\times((\xi/||\xi||\times E_{0})=0.
\]
Since $\xi/||\xi||$ is a unitary vector, we have that the second
term in the left hand side of the above equation is just the negative
of the projection of $E_{0}$ onto the orthogonal space to $\xi$,
and hence
\[
(\tau/||\xi||)^{2}<\epsilon E_{0},v>+<E_{0},v>=0\hbox{\,\, for\, all\,\,}v\hbox{\, shuch\, that\,}<v,\xi>=0.
\]
The above equation means that 
\[
(\sigma_{\epsilon}(\xi/||\xi||)+(||\xi||/\tau)^{2}I)(\pi_{\xi}E_{0})=0,
\]
where $\pi_{\xi}:\RR^{3}\rightarrow\tangent_{\xi/||\xi||}S^{2}$ is
projection onto the orthogonal complement of $\xi$, and $I$ is the
identity morphism of the bundle $\sym{(\tangent S^{2})}$. Hence,
we must have that 
\[
\tau=||\xi||\lambda_{\epsilon,i}^{-1/2}(\xi/||\xi||)
\]
where $\lambda_{i}:S^{2}\rightarrow\RR$ are the eigenvalue functions
of $\sigma_{\epsilon}$. From this formula and the definition of $\FH_{\epsilon}$
(see \ref{eq:FresnelSurface}) we obtain the desired result.\end{proof}
\begin{example}
Consider the case were $\epsilon$ is diagonal with eigenvalues $\epsilon_{1}=1,\epsilon_{2}=2$
and $\epsilon_{3}=15$, as shown in Figure \ref{fig:Biaxial}. The
surface at the bottom right corresponds to $\lambda_{\text{\ensuremath{\epsilon},1 }}$
and the one to the left $\lambda_{\epsilon,2}$. Observe that that
these surfaces join at the points $\xi\in S^{2}$ where $\lambda_{\epsilon,1}(\xi)=\lambda_{\epsilon,2}(\xi)$,
and exactly at these points are where the singularities of $\FH_{\epsilon}$
occur. 
\end{example}
The previous example motivates the importance of the following definition.
\begin{defn}
If $s$ is a section of $\sym F$ we will say that $p\in X$ is a
\emph{multiple point} iff $\lambda_{s,1}(p)=\lambda_{s,2}(p)$, and
we will denote the set of such points by $\mult_{s}$. 
\end{defn}
Let $\symz F$ stand for traceless elements in $\sym F$. For a given
section $s$ of $\sym F$ we can construct its traceless part (which
is a section of $\symz F$) by letting
\[
s_{0}=s-\frac{1}{2}\tr(s)I_{F},
\]
where $I_{F}$ is the identity section in $\sym F$. 
\begin{prop}
\label{prop:MasZeros}The set $\mult_{s}$ coincides with the set
of zeros of the section $s_{0}$. Furthermore, if $s_{0}$ is transversal
to the zero section of $\symz F$ then the eigenvalue function $\lambda_{s,1}$
and $\lambda_{s,2}$ have conic singularities on $\mult_{s}$.\end{prop}
\begin{proof}
It is easy to see that for a section $s$ of $\sym F$ we have that
\begin{eqnarray*}
\lambda_{s,1}(p) & = & \frac{1}{2}\tr(s(p))-||s_{0}(p)||\\
\lambda_{s,2}(p) & = & \frac{1}{2}\tr(s(p))+||s_{0}(p)||.
\end{eqnarray*}
where the norm used for an element $A\in\symz F_{p}$ is given by
$||A||=\frac{1}{2}\tr(A^{2})$. From this it follows that $\lambda_{s,1}(p)=\lambda_{s,2}(p)$
iff $s_{0}(p)=0$. If $s_{0}$ is transversal to the zero section
of $\symz F$ then near a zero of $s_{0}$ the norm function behaves
like $(x,y)\map(x^{2}+y^{2})^{1/2}$, and hence we have conical singularities
at these points.
\end{proof}
We are now ready to prove that four singularities of $\FH_{\epsilon}$
can be accounted from purely topological considerations. We will need
the following result first.
\begin{lem}
\label{lem:S0FtoFF}The bundle $\symz F$ is isomorphic to $F\otimes_{\CC}F$,
as a real vector bundle. \end{lem}
\begin{proof}
Since $F$ has a metric, we can consider it as an $\SO(2)$ bundle.
For 
\[
R=\left(\begin{array}{cc}
\cos(\theta) & -\sin(\theta)\\
\sin(\theta) & \cos(\theta)
\end{array}\right)\in\SO(2)
\]
and
\[
A=\left(\begin{array}{cc}
a & b\\
b & -a
\end{array}\right)\in\symz{\RR^{2}}
\]
we have that 
\[
RAR^{T}=\left(\begin{array}{cc}
p & q\\
q & -p
\end{array}\right),
\]
where
\[
\left(\begin{array}{c}
p\\
q
\end{array}\right)=\left(\begin{array}{cc}
\cos(2\theta) & -\sin(2\theta)\\
\sin(2\theta) & \cos(2\theta)
\end{array}\right)\left(\begin{array}{c}
a\\
b
\end{array}\right).
\]
Hence, the transition functions of the bundle $\symz F$ are the same
as those of the bundle $F\otimes_{\CC}F$.\end{proof}
\begin{prop}
If $s$ is a section of $\sym F$ such that $s_{0}$ is transversal
to the zero section of $\symz F$, then the cardinality of $\mult_{s}$
is at least $2\left|\int_{X}e(F)\right|$.\end{prop}
\begin{proof}
From Proposition \ref{prop:MasZeros} we have that the $\mult_{s}=Z_{0}$,
where $Z_{0}$ is the set of zeros of $s_{0}$. We have that
\[
\int_{X}e(\symz F)=\sum_{p\in Z_{0}}i_{s_{0}}(p),
\]
where $i_{s_{0}}(p)$ is the index of $s_{0}$ at $p$. From the transversality
assumption we have that $i(p)=\pm1$, and from Lemma \ref{lem:S0FtoFF}
we have that $e(\symz F)=e(F\otimes_{\CC}F)=2e(F)$. We conclude that
\[
|Z_{0}|\geq2\left|\int e(F)\right|,
\]
where $|Z_{0}|$ is the cardinality of $Z_{0}$. \end{proof}
\begin{cor}
If $X$ has genus $g$ and $s:X\rightarrow\sym{(\tangent X)}$ is
such that $s_{0}$ is transversal to the zero section of $\symz{(\tangent X)},$
then cardinality of $\mult_{s}$ is at least $4-4g$.\end{cor}
\begin{proof}
Apply the above proposition and the formula
\[
\int_{X}e(\tangent X)=2-2g.
\]

\end{proof}
We know that the singularities of $\FH_{\epsilon}$ occur at the multiplicity
set $\mult_{\epsilon}$ of the section $s_{\epsilon}$. And the from
the above Corollary we know that the cardinality of $\mult_{\epsilon}$
must be greater or equal than $4$. This is accordance with the known
fact that $\FH_{\epsilon}$ has exactly four singularities, but our
argument has been purely topological. In this discussion we have implicitly
assumed that $(s_{\epsilon})_{0}$ is transversal to the zero section
of $\sym{(\tangent S^{2}})$. We can explicitly find the points in
$\mult_{\epsilon}$ and verify that this is actually the case. 
\begin{prop}
\label{prop:s_epsilon_0_transversal}If the the eigenvalues of $\epsilon\in\sym{\RR^{3}}$
are all different from each other, then the traceless part of $s_{\epsilon}$
is transversal to the zero section of $\symz{(\tangent S^{2})}$.\end{prop}
\begin{proof}
We can assume that $\epsilon$ is diagonal with diagonal elements
$0<\epsilon_{1}<\epsilon_{2}<\epsilon_{3}$. The zero set of $(s_{\epsilon})_{0}$
does not contain neither $(0,0,1)$ nor $(0,0,-1)$. We can use polar
coordinates

\[
p(\theta,\varphi)=(\cos(\theta)\cos(\varphi),\sin(\theta)\cos(\varphi),\sin(\varphi))\hbox{\,\, for\,\,}0\leq\theta<2\pi,-\pi/2<\varphi<\pi/2,
\]
and let
\begin{eqnarray*}
u(\theta,\varphi) & = & (-\sin(\theta),\cos(\theta),0)\\
v(\theta,\varphi) & = & (-\cos(\text{\ensuremath{\theta})\ensuremath{\sin}(\ensuremath{\varphi}),-\ensuremath{\sin}(\ensuremath{\theta})\ensuremath{\sin}(\ensuremath{\theta}),\ensuremath{\cos}(\ensuremath{\varphi})). }
\end{eqnarray*}
The triple $(p,u,v)$ forms and orthonormal frame, so that locally
we can write
\[
s_{\epsilon}=\left(\begin{array}{cc}
<\epsilon u,u> & <\epsilon u,v>\\
<\epsilon v,u> & <\epsilon v,v>
\end{array}\right)
\]
If we identify $\symz{\RR^{2}}$ with $\RR^{2}$ by letting
\[
(a,b)\sim\left(\begin{array}{cc}
a & b\\
b & -a
\end{array}\right)
\]
then we can write
\[
(s_{\epsilon})_{0}=\frac{1}{2}\left(<\epsilon u,u>-<\epsilon v,v>,2<\epsilon u,v>\right)
\]
where
\begin{eqnarray*}
2<\epsilon u,v> & = & 2(\epsilon_{1}-\epsilon_{2})\cos(\theta)\sin(\theta)\sin(\varphi),\\
<\epsilon u,u>-<\epsilon v,v> & = & (\epsilon_{2}-\epsilon_{1}\sin^{2}(\varphi))\cos^{2}(\theta)+\\
 & + & (\epsilon_{1}-\epsilon_{2}\sin^{2}(\varphi))\sin^{2}(\theta)-\epsilon_{3}\cos^{2}(\varphi).
\end{eqnarray*}
From this formula it is easy to see that the zeros of $(s_{\epsilon})_{0}$
occur at the points $(0,\varphi_{m}),(0,\varphi_{m}+\pi/2),(\pi,\varphi_{m})$
and $(0,\varphi_{m}+\pi/2)$, where 
\[
\varphi_{m}=\arctan\left(\frac{\epsilon_{3}-\epsilon_{2}}{\epsilon_{2}-\epsilon_{1}}\right)^{\frac{1}{2}}.
\]
A simple calculations shows that
\[
D(s_{\epsilon})_{0}(0,\varphi_{m})=(\epsilon_{2}-\epsilon_{1})^{3/2}(\epsilon_{3}-\epsilon_{2})(\epsilon_{3}-\epsilon_{1})^{-1/2}>0,
\]
and hence $(s_{\epsilon})_{0}$ is transversal to the zero section
of $\sym{(\tangent S^{2})}$ at that point. Due to the symmetry of
the problem, the same condition holds at the other points.
\end{proof}

\section{Desingularising the Fresnel Surface and Morse Theory}

We are interested in finding the critical points of the functions
$\lambda_{\epsilon,1}$ and $\lambda_{\epsilon,2}$ described in the
previous section. One reason why this problem is interesting is that
the maximum and minimum are critical points of these functions, and
they correspond to directions of maximal and minimal wave velocities.
We want use a topological argument to find lower bounds to the number
and types of critical points. To do this, we can apply Morse theory.
Recall that if $f:X\rightarrow\RR$ is a smooth function and $p\in X$
is a critical point of $f$ (i.e $df(p)=0$) then we can define its
hessian $d^{2}f(p)$, which is symmetric bilinear form in $\tangent_{p}X$.
A critical point $p\in X$ of $f$ is said to be non-degenerate iff
$d^{2}f(p)$ is non-degenerate, and in this case we define the index
of $p$ as the dimension of the maximal subspace of $\tangent_{p}X$
at which $d^{2}f(p)$ is negative definite. Morse inequalities assert
that if $f$ has only non-degenerate critical points the we have that
(see \cite[pg. 29]{kn:morse})
\begin{eqnarray*}
C_{i}(f) & \geq & \dim(H^{i}(X,\RR))\\
\sum_{i=0}^{\dim(X)}(-1)^{i}\dim(H^{i}(X,\RR)) & = & \sum_{i=0}^{\dim(X)}(-1)^{i}C_{i}(f),
\end{eqnarray*}
where $C_{i}(f)$ are the number of critical points of $f$ of index
$i$, and $H^{i}(X,\RR)$ is the cohomology group of $X$ with real
coefficients. This formulas allow us to estimate the number of critical
points of a given index in terms of topological invariants of $X$. 

We have seen in the previous section both $\lambda_{\epsilon,1}$
and $\lambda_{\epsilon,2}$ are non-smooth at $\mult_{\epsilon}$,
so we cannot apply Morse Theory directly to them. To solve this problem
we desingularise these functions as follows. For a section $s:X\rightarrow\sym F$
we define 
\[
\EE_{s}=\bigcup_{p\in X}\{l_{x}\in\proj F|l_{p}\hbox{\, is\, spanned\, by\, an\, eigenvector\, of\,\,}s(p)\},
\]
where $\proj F$ is the proyectivisation of $F$. We will refer to
$\EE_{s}$ as the \emph{space of eigenlines} of $s$. Over this space
we can define the eigenvalue function of $s$ by 
\[
\lambda_{s}(l_{p})=\hbox{eigenvalue\, of\,\,}s(p)\hbox{\, corresponding\, to\,\,}l_{p}.
\]
It turns out that if $s_{0}$ is transversal to the zero section of
$\symz F$ then we have that $\EE_{s}$ is a smooth submanifold of
$\proj F$, and $\lambda_{s}$ is a smooth function (see \cite[Theorem 5]{kn:mocs}).
Furthermore, the space $\EE_{s}$ is topologically equivalent to the
space obtained by the following procedure (see \cite[Proposition 6]{kn:mocs})
\begin{enumerate}
\item Let $X_{0}$ be obtained by removing small open disks, centred at
the points of $\mult_{s}$, from $X.$
\item The space $\EE_{s}$ is homeomorphic the the space obtained by joining
two disjoint copies of $X_{0}$ along cylinders of the form $S^{1}\times(-\delta,\delta)$
(see Figure \ref{fig:Space of Eigenlines}). 
\end{enumerate}
\begin{figure}

\includegraphics[scale=0.8]{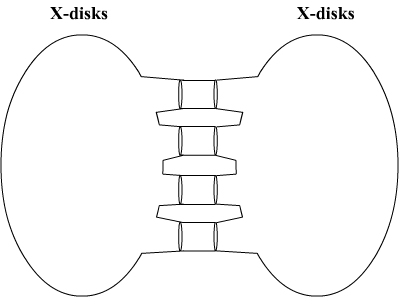}\caption{\label{fig:Space of Eigenlines}The space $\EE_{s}$}

\end{figure}

If we restrict the projection map $\pi:\proj F\rightarrow X$ to $\EE_{s}$
we obtain a projection map $\pi_{s}:\EE_{s}\rightarrow X$ with the
following properties (see \cite[Theorem 5]{kn:mocs})
\begin{enumerate}
\item We have that $\pi_{s}^{-1}(X-\mult_{s})$ consists of two connected
component $X_{1}$ and $X_{2}$ such that $\pi_{s}:X_{i}\rightarrow X-\mult_{s}$
is a diffeomorphism and $\lambda_{s,i}\circ\pi_{s}=\lambda_{s}$.
\item For any $p\in\mult_{s}$ we have that $\pi_{s}^{-1}(p)$ is diffeomorphic
to a circle and $\lambda_{s}$ can have $0,1$ or $2$ critical points
at $\pi^{-1}(p)$. Furthermore, if any of these critical points is
non-degenerate it must have index equal to one.
\end{enumerate}
To simplify notation we define
\[
\EE_{\epsilon}=\EE_{s_{\epsilon}},\pi_{\epsilon}=\pi_{s_{\epsilon}}\hbox{\, and\,}\lambda_{\epsilon}=\lambda_{s_{\epsilon}}.
\]
As an application of the above results, we have the following.
\begin{prop}
If $\epsilon\in\sym{\RR^{3}}$ has eigenvalues $0<\epsilon_{1}<\epsilon_{2}<\epsilon_{3}$
then then space $\EE_{\epsilon}$ is a smooth surface of genus $3$
and $\lambda_{\epsilon}:\EE_{\epsilon}\rightarrow\RR$ is a smooth
function having no critical points at $\pi_{\epsilon}^{-1}(\mult_{\epsilon}).$\end{prop}
\begin{proof}
The spaces $\EE_{\epsilon}$ is obtained joining two disjoint copies
of $S^{2}-D_{1}\cup D_{2}\cup D_{3}\cup D_{4}$ along cylinders of
the form $S^{1}\times(-\delta,\delta)$, where the $D_{i}$\textasciiacute{}s
are small open disks centred at the four points in $\mult_{\epsilon}$
(see Figure \ref{fig:Space of Eigenlines}). From this it follows
that $\EE_{\epsilon}$ is a surface of genus $3$. To check that $\lambda_{\epsilon}$
has no singularities in $\pi_{\epsilon}^{-1}(\mult_{\sigma})$ we
write a explicit formula for $\lambda_{\epsilon}$ and check that
this is actually the case for $0<\epsilon_{1}<\epsilon_{2}<\epsilon_{3}$.
Consider polar coordinates on $S^{2}-\{(0,0,1),(0,0,-1)$ given by
\[
p(\theta,\varphi)=(\cos(\theta)\cos(\varphi),\sin(\theta)\cos(\varphi),\sin(\varphi))
\]
and the orthonormal frame $\{u(\theta,\varphi),v(\theta,\varphi)\}$
in $\tangent(S^{2}-\{(0,0,1),(0,0,-1))$ given by
\begin{eqnarray*}
u(\theta,\varphi) & = & (-\sin(\theta),\cos(\theta),0)\\
v(\theta,\varphi) & = & (-\cos(\text{\ensuremath{\theta})\ensuremath{\sin}(\ensuremath{\varphi}),-\ensuremath{\sin}(\ensuremath{\theta})\ensuremath{\sin}(\ensuremath{\theta}),\ensuremath{\cos}(\ensuremath{\varphi})). }
\end{eqnarray*}
If we let $\alpha$ be the angle of a line $l\in\proj{(\tangent S^{2})}$
with respect to $u(\theta,\varphi)$ then we can use $(\theta,\varphi,\alpha)$
as coordinates for $\proj{(\tangent S^{2})}$. In these coordinates
we have that (see \cite[proof of Theorem 5 ]{kn:mocs}
\begin{eqnarray*}
\EE_{\epsilon} & = & f_{\epsilon}^{-1}(0)\\
\lambda_{\epsilon} & = & g_{\epsilon}|\EE_{\epsilon}
\end{eqnarray*}
where 
\begin{eqnarray*}
f_{\epsilon}(\theta,\varphi,\alpha) & = & (\epsilon_{2}-\epsilon_{1})\cos(2\alpha)\cos(\text{\ensuremath{\theta})\ensuremath{\sin}(\ensuremath{\theta})\ensuremath{\sin}(}\varphi)\\
 & + & \frac{1}{2}\sin(2\alpha)(\cos^{2}(\theta)(\epsilon_{2}-\epsilon_{1}\sin^{2}(\varphi))\\
 & + & \frac{1}{2}\sin(2\alpha)(\sin^{2}(\theta)(\epsilon_{1}-\epsilon_{2}\sin^{2}(\varphi))-\epsilon_{3}\cos^{2}(\varphi))\\
g_{\epsilon}(\theta,\varphi,\alpha) & = & -\sin(2\alpha)\sin(\theta)\sin(\varphi)\\
 & + & \frac{1}{2}(\epsilon_{1}\sin^{2}(\theta)+\epsilon_{2}\cos^{2}(\theta)+\epsilon_{3}\cos^{2}(\varphi)))\\
 & + & \frac{1}{2}(\epsilon_{1}\cos^{2}(\theta)+\epsilon_{2}\sin^{2}(\theta))\sin^{2}(\varphi)\\
 & + & \frac{1}{2}\cos(2\alpha)(\cos^{2}(\theta)(\epsilon_{2}-\epsilon_{1}\sin^{2}(\varphi))\\
 & + & \frac{2}{3}\cos(2\alpha)(\sin^{2}(\theta)(\epsilon_{1}-\epsilon_{2}\sin^{2}(\varphi)-\epsilon_{\text{3}}\cos^{2}(\varphi))
\end{eqnarray*}
From this, we obtain that
\begin{eqnarray*}
\nabla f_{\epsilon}(0,\varphi_{m},\alpha) & = & \left(\frac{(\epsilon_{1}-\epsilon_{2})(\epsilon_{3}-\epsilon_{2})^{1/2}}{(\epsilon_{3}-\epsilon_{1})^{1/2}}\cos(2\alpha),(\epsilon_{2}-\epsilon_{1})^{1/2}(\epsilon_{3}-\epsilon_{2})^{1/2}\sin(2\alpha),0\right)\\
\nabla g_{\epsilon}(0,\varphi_{m},\alpha) & = & \left(\frac{(\epsilon_{3}-\epsilon_{2})^{1/2}}{(\epsilon_{3}-\epsilon_{1})^{1/2}}\sin(2\alpha),(\epsilon_{2}-\epsilon_{1})^{1/2}(\epsilon_{3}-\epsilon_{2})^{1/2}(1-\cos(2\alpha),0\right).
\end{eqnarray*}
where $(0,\varphi_{m})$ corresponds to one of the zeros of $(s_{\epsilon})_{0}$
as in the proof of Proposition \ref{prop:s_epsilon_0_transversal}
(by symmetry the other zeros of $(s_{\epsilon})_{0}$ are dealt in
the same way). The critical points of $\lambda_{\epsilon}$ correspond
to the $\alpha$\textasciiacute{}s at which $\nabla f_{\epsilon}(0,\varphi_{m},\alpha)=\nabla g_{\epsilon}(0,\varphi_{m},\alpha)$.
By using the above formulas we have that this equality can not hold
if all the $\epsilon_{i}$\textasciiacute{}s are different from each
other. \end{proof}
\begin{cor}
If $\epsilon\in\sym{\RR^{3}}$ has eigenvalues $0<\epsilon_{1}<\epsilon_{2}<\epsilon_{3}$
then we have that $C_{i}(\lambda_{\text{\ensuremath{\epsilon}}})=C_{i}(\lambda_{\epsilon,1})+C_{i}(\lambda_{\epsilon,2})$
for $i=0,1,2$.\end{cor}
\begin{proof}
By the previous proposition we have that $\lambda_{\epsilon}$ has
no critical points on $\pi_{\epsilon}^{-1}(\mult_{\epsilon})$. But
we know that on each of the two connected components of $\EE_{\epsilon}-\pi_{\epsilon}^{-1}(\mult_{\epsilon})$
the function $\lambda_{\epsilon}$ coincides (up to diffeomorphism)
with $\lambda_{\epsilon,1}$or $\lambda_{\epsilon,2}$.
\end{proof}
Applying Morse inequalities to $\lambda_{\epsilon}$ and using the
above results we obtain that
\begin{eqnarray*}
C_{0}(\lambda_{\epsilon,1})+C_{0}(\lambda_{\epsilon,2}) & \geq & 1\\
C_{1}(\lambda_{\epsilon,1})+C_{1}(\lambda_{\epsilon,2}) & \geq & 3\\
C_{2}(\lambda_{\epsilon,1})+C_{2}(\lambda_{\epsilon,2}) & \geq & 1\\
\sum_{i=0}^{2}C_{0}(\lambda_{\epsilon,1})+C_{0}(\lambda_{\epsilon,2}) & = & -4
\end{eqnarray*}
The above holds under the assumption that $\lambda_{\epsilon}$ is
a Morse function, but it can easily be seen that is the case if the
eigenvalues of $\epsilon$ are all different from each other.

\section{Generalisations - Future Work}

In this section we propose a program to generalise our results to
general hyperbolic differential operators on bundles. 

Let $E$ and $F$ be vector bundles of rank $m$ and $k$ over a Riemannian
manifold $X$ of dimension $n$. Consider a linear differential operator
of degree $d\in\ZZ^{+}$ of the form
\begin{equation}
L=\sum_{i=0}^{d}\left(L_{i}\circ\frac{\partial^{d-i}}{\partial t^{d-i}}\right),\label{eq:MainDiffOperator}
\end{equation}
where $L_{i}:C^{\infty}(E)\rightarrow C^{\infty}(F)$ is a differential
operator of order $i$. Observe that $L$ acts on time dependent sections
of $E$. We are interested in \emph{high frequency solutions} of the
equation $Lu=0$, i.e in solutions to the asymptotic partial differential
equation 
\[
L\left(u_{0}e^{is\varphi}\right)=0\hbox{\, as\,\,}s\mapsto\infty,
\]
where $u_{0}$ is a section of $E$ and $\varphi$ is a smooth real
valued function in $X\times\RR$ (known in the physics literature
as the \emph{phase function}). The above asymptotic equation leads
to the the equation (see \cite[pg. 31]{kn:guillemin})
\begin{equation}
\left(\sum_{i=0}^{d}\sigma_{i}(d_{x}\varphi(x,t))\left(\frac{\partial\varphi}{\partial t}(x,t)\right)^{d-i}\right)u_{0}(x)=0,\label{eq:PolarisationEquation}
\end{equation}
where $\sigma_{i}:\cotangent X\rightarrow\hbox{Hom}(E,F)$ is the
\emph{principal symbol} of $L_{i}$. If we choose local coordinates
for $X$ and local trivialisation of $E$ and $F$ over an open set
$U\subset X,$ we can write
\[
L_{i}=\sum_{|\alpha|\leq i}A_{i,\alpha}\frac{\partial^{\alpha}}{\partial x^{\alpha}}
\]
where
\[
\alpha=(\alpha_{1},\ldots,\alpha_{n}),|\alpha|=\alpha_{1}+\cdots+\alpha_{n}\hbox{\, and\,\,\,}\frac{\partial^{\alpha}}{\partial x^{\alpha}}=\frac{\partial^{\alpha_{1}+\cdots+\alpha_{n}}}{(\partial x_{1})^{\alpha_{1}}\cdots(\partial x_{n})^{\alpha_{n}}}
\]
and for every $\alpha$ and $i$ we have that $A_{i,\alpha}$ is a
function from $U$ to the $k\times m$ matrices with coefficients
in $\RR$. The local expression for $\sigma_{i}$ is then given by
\[
\sigma_{i}(x,\xi)=\sum_{|\alpha|=i}A_{i,\alpha}(x)\xi^{\alpha}\hbox{\,\, where\,\,\,}\xi=\xi_{1}^{\alpha_{1}}\xi_{2}^{\alpha_{2}}\cdots\xi_{n}^{\alpha_{n}}.
\]
We see that over the fibre variables of $\cotangent X$ the symbol
$\sigma_{i}$ is an homogenous polynomial of degree $i$ whose coefficients
are real $k\times m$ matrices. 

We define the\emph{ Fresnel hypersurface }associated to $L$ by 
\begin{equation}
\FH_{L}=\left\{ (x,-\tau\xi)\in\cotangent X\left|(x,\xi)\in S(\cotangent X)\hbox{\, and\,}\ker\left(\sum_{i=0}^{d}\sigma_{i}(x,\xi)\tau^{d-i}\right)\not=0\right.\right\} ,\label{eq:characteristicSet}
\end{equation}
We will assume that the following conditions hold
\begin{enumerate}
\item We have that $F$ is a Riemannian real vector bundle and $E=F$ .
For example, if we were modelling small oscillations of an elastic
surface $X$ then we could assume that $E=F=\tangent X$.
\item For all $(x,\xi)\in\cotangent X$ and $1\leq i\leq d$ we have that
$\sigma_{i}(x,\xi)$ is a symmetric operator with respect to a metric
in $F$, i.e
\[
<\sigma_{i}(x,\xi)v,w>_{F}=<v,\sigma_{i}(x,\xi)w>_{F}\hbox{\,\, for\, all\,}v,w\in F_{x}.
\]

\end{enumerate}
It turns out that two conditions above hold in many cases of physical
interest, and in particular those physical problems that arise from
variational principles. Under the above assumptions we can write
\[
\FH_{L}=\{(x,-\tau\xi)\in S(\cotangent X)|p_{L}(x,\xi,\tau)=0\},
\]
where $S(\cotangent X)$ are the unit vectors in $\cotangent X$ and
\[
p_{L}(x,\xi,\tau)=\det\left(\sum_{i=0}^{d}\sigma_{i}(x,\xi)\tau^{d-i}\right).
\]

\begin{defn}
A polynomial $p=p(\tau)$ is said to be \emph{hyperbolic} if it has
as many real roots as its degree (counting multiplicities) and it
is said to be \emph{strictly hyperbolic} if all these roots are different
from each other. 
\end{defn}
For a fixed $(x,\xi)\in S(\cotangent X)$ let
\[
p_{L,x,\text{\ensuremath{\xi}}}(\tau)=p_{L}(x,\xi,\tau).
\]

\begin{defn}
The operator $L$ is said to be hyperbolic if $p_{L,x,\xi}$ is hyperbolic
for all $(x,\xi)\not=0$ in $\cotangent X$.
\end{defn}
If $L$ is hyperbolic we have functions $\lambda_{i}:\cotangent X\rightarrow\RR$
for $1\leq i\leq\dim(F)\cdot\hbox{deg}(L)$ given as as the roots
of $p_{L,x,\xi}$, and we can write
\[
\FH_{L}=\bigcup_{i=1}^{d}\{-\lambda_{i}(\xi)\xi|\xi\in S(\cotangent X)\}.
\]
We define $\mult_{L}$ as the set of points in $S(\cotangent X)$
at which two of more or the $\lambda_{i}$\textasciiacute{}s coincide.
The singularities of $\FH_{L}$ occur at points on the set $\mult_{L}$.
Inspired by our work on Maxwell\textasciiacute{}s equations, we pose
the following problems.
\begin{enumerate}
\item \emph{Find topological obstructions to the condition }$\mult_{L}=\emptyset$.
This will provide obstructions to the existence of strictly hyperbolic
differential operators on $F$. It is expected that the characteristic
classes of $F$ should be involved in the answer to this problem.
\item \emph{Desingularisation of} $\FH_{L}$. As in the case of Maxwell\textasciiacute{}s
equation, we would like to find a smooth space $\EE_{L}$ and a smooth
function $\lambda_{L}:\EE_{L}\rightarrow\RR$ that contains all the
information of the $\lambda_{L,i}$\textasciiacute{}s. We would then
expect to be able to apply Morse theory to $\lambda_{L}$ to obtain
similar results as in the case of Maxwell\textasciiacute{}s equations.
\end{enumerate}
In relation with Problem 1 above we mention the work in the articles
\cite{kn:braam,kn:hormander,kn:john,kn:khesin,kn:lax,kn:nuij}, which
contain some local and global results related to the multiplicity
of eigenvalues of symbols of differential and pseudo-differential
operators. Regarding Problem 2, we mention that the non-smoothness
of $\lambda_{L,i}$\textasciiacute{}s at $\mult_{L}$ explain phenomena
like wave transformation (see \cite[pg. 223]{kn:arnold}) in which
longitudinal waves transforms into transversal ones (or viceversa).

\bibliographystyle{plain}
\bibliography{myBib}

\end{document}